\documentclass{article}

\usepackage{microtype}
\usepackage{graphicx}
\usepackage{subcaption}
\usepackage{booktabs} 

\usepackage{hyperref}

\usepackage[preprint]{icml2026}

\usepackage{amsmath}
\usepackage{amssymb}
\usepackage{mathtools}
\usepackage{amsthm}
\usepackage{nicefrac}
\usepackage{makecell}

\usepackage[capitalize,noabbrev]{cleveref}

\theoremstyle{plain}
\newtheorem{theorem}{Theorem}[section]

\newtheorem{lemma}[theorem]{Lemma}
\newtheorem{corollary}[theorem]{Corollary}
\newtheorem{fact}[theorem]{Fact}

\theoremstyle{definition}
\newtheorem{definition}[theorem]{Definition}

\theoremstyle{remark}

\usepackage{cleveref}
\usepackage{thm-restate}

\newcommand{\cA}{\mathcal{A}}
\newcommand{\cB}{\mathcal{B}}

\newcommand{\cM}{\mathcal{M}}

\newcommand{\bbM}{\mathbb{M}}
\newcommand{\bbN}{\mathbb{N}}
\newcommand{\bbR}{\mathbb{R}}

\newcommand{\R}{\mathbb{R}}

\newcommand{\Opt}{\mathsf{OPT}}

\newcommand{\calA}{\mathcal{A}}
\newcommand{\calB}{\mathcal{B}}
\newcommand{\calC}{\mathcal{C}}
\newcommand{\calM}{\mathcal{M}}

\newcommand{\OPT}{\mathsf{OPT}}

\newcommand{\floor}[1]{\lfloor #1 \rfloor}

\DeclareMathOperator{\E}{\mathbb{E}}
\DeclareMathOperator{\cost}{\mathsf{cost}}
\DeclareMathOperator{\dist}{\mathsf{dist}}

\usepackage[textsize=tiny]{todonotes}

\begin{document}

\twocolumn[
  \icmltitle{Parsimonious Learning-Augmented Online Metric Matching}



  \icmlsetsymbol{equal}{*}

  \begin{icmlauthorlist}
    \icmlauthor{Yongho Shin}{equal,uwr}
    \icmlauthor{Phanu Vajanopath}{equal,uwr}
  \end{icmlauthorlist}

  \icmlaffiliation{uwr}{Institute of Computer Science, University of Wrocław, Wrocław, Poland}

  \icmlcorrespondingauthor{Yongho Shin}{yongho@cs.uni.wroc.pl}

  \icmlkeywords{learning-augmented algorithms, online metric matching, parsimonious predictions}

  \vskip 0.3in
]



\printAffiliationsAndNotice{\icmlEqualContribution}

\begin{abstract}
  \emph{Learning-augmented algorithms} have received significant attention in recent years, particularly in the context of online optimization.
Motivated by the high computational cost of generating predictions, a growing line of work studies the tradeoff between performance guarantees and the number of predictions used in learning-augmented algorithms for problems such as caching and metrical task systems.
In this paper, we extend this line of research to \emph{online metric matching} by developing \emph{parsimonious} learning-augmented algorithms and establishing lower bounds on their performance. 
Our approach extends the Follow-the-Prediction framework to the parsimonious setting by filling in a \emph{virtual prediction} in the absence of an actual prediction, using an online metric matching algorithm that maintains good intermediate matchings throughout its execution.
We complement our theoretical results with an empirical evaluation, demonstrating the practical effectiveness of our approach.
\end{abstract}

\section{Introduction} \label{sec:intro}
\emph{Online metric matching} is a fundamental online optimization problem that naturally arises in applications such as ride-hailing and resource allocation \cite{kalyanasundaram1993online,khuller1994line}. 
The problem is defined as follows. 
We are given a metric containing $n$ servers and $n$ requests.
The locations of the servers are known in advance, while the requests arrive
one at a time step in an online fashion. 
Upon the arrival of a request, the algorithm must irrevocably match it to an available server at the moment; once matched, a server cannot be used again, and a request cannot be reassigned to another server at a later time step. 
At termination, the algorithm produces a perfect matching between servers and requests. 
The objective is to minimize the cost of the perfect matching, measured as the sum of distances between matched pairs.

Due to a dedicated line of work~\cite{kalyanasundaram1993online, khuller1994line, meyerson2006randomized, bansal2014randomized}, this problem is fairly well understood.
On the algorithmic side, there exists a deterministic algorithm that outputs a perfect matching whose cost is at most a factor of $2n-1$ larger than that of an optimal perfect matching with hindsight~\cite{kalyanasundaram1993online,khuller1994line}.
Later, \citet{meyerson2006randomized} gives a randomized algorithm returning a solution of expected cost within a factor of $O(\log^3 n)$ to the optimal cost with hindsight;
this factor was further improved to $O(\log^2 n)$ by \citet{bansal2014randomized}.
It is also known that these guarantees are (nearly) best-possible. 
In particular, no deterministic online algorithm can achieve a factor better than $2n-1$~\cite{kalyanasundaram1993online,khuller1994line}, matching the above algorithmic result. For randomized algorithms, an $\Omega(\log n)$ lower bound is known~\cite{meyerson2006randomized}.

A major source of these strong lower bounds is the uncertainty about the future input. 
However, recent advances in AI and machine learning suggest that reasonable predictions about the future input may be available in many practical settings.
Motivated by this observation, \emph{learning-augmented algorithms}, also known as \emph{algorithms with predictions}, have received significant attention in recent years, especially in online optimization since the seminal work of ~\citet{lykouris2021competitive}. These algorithms are provided with predictions as part of the input and are designed to leverage them to solve the problem. 
The performance is typically measured by the following three factors~\cite{kumar2018improving}: \emph{consistency}, which measures performance when predictions are accurate; \emph{robustness}, which guarantees performance regardless of prediction quality; and \emph{smoothness}, which captures how gracefully performance degrades as prediction error increases.

Online metric matching has also been studied in this learning-augmented framework. \citet{antoniadis2023online} introduced the \emph{action prediction} model and presented a deterministic algorithm that uses the prediction provided at every time and outputs a perfect matching with cost at most
$9\cdot\min \{  \cost(\OPT)+2\eta), (2n-1) \cost(\OPT) \}$,\footnote{This expression reflects the correct performance guarantee derivable from the analysis of \citet{antoniadis2023online}, which is inconsistent with the expression reported in the paper \cite{coester2026}.} where $\cost(\OPT)$ denotes the cost of an optimal perfect matching with hindsight and $\eta$ denotes the error of action prediction.

While predictions can substantially improve performance, generating them is often computationally expensive. 
This motivates the study of learning-augmented algorithms that use predictions parsimoniously. 
\citet{im2022parsimonious} initiated this line of work by studying the tradeoff between performance guarantee and a bounded number of predictions in online caching.
Subsequently, \citet{sadek2024algorithms} introduced a more restricted regime in which two predictions must be separated by a minimum time gap, and analyzed the tradeoffs for metrical task systems and caching. 
In this context, \citet{im2022parsimonious} posed a question of investigating a parsimonious prediction model for online metric matching.

\paragraph{Our Contributions}
We answer the question posed by \citet{im2022parsimonious} by presenting parsimonious learning-augmented algorithms for online metric matching under action predictions, together with lower bounds on their performance.

We first present our algorithmic results for general metric spaces. 
We denote by $\cost(\OPT)$ the cost of an optimal perfect matching with hindsight, and by $\eta$ the total error incurred by the actual predictions used by the algorithm.
\begin{theorem} \label{thm:intro:det}
    For any $k \in \bbN$, there exists a deterministic learning-augmented algorithm for online metric matching that outputs a perfect matching with cost at most
    \[
        9 \cdot \min \bigg\{ 
        \begin{array}{l}
            (2k-1)\cost(\OPT)+2k\eta, \\[2pt]
            (2n-1)\cost(\OPT) 
        \end{array}
        \bigg\},
    \]
    where any two predictions used by the algorithm are separated by at least $k$ time steps.
\end{theorem}

\begin{theorem} \label{thm:intro:rand}
    For any $k \in \bbN$, there exists a randomized learning-augmented algorithm for online metric matching that outputs a perfect matching with expected cost at most
    \[
        \min \bigg\{ 
        \begin{array}{l}
            O(\log n \log k)\cdot(\cost(\OPT)+\eta), \\[2pt]
            O(\log^2 n)\cdot \cost(\OPT)
        \end{array}
        \bigg\},
    \]
    where any two predictions used are deterministically separated by at least $k$ time steps.
\end{theorem}

In addition, we obtain improved guarantees for specific metric spaces, including line metrics and hierarchically well-separated trees (see \Cref{sec:alg:apply}).

Our algorithms are built upon a meta-algorithm that can be viewed as a parsimonious variant of the Follow-the-Prediction framework of \citet{antoniadis2023online}. 
A key challenge is that this framework requires a prediction at every time step, whereas in our algorithms, predictions can only be used sparingly with a minimum separation of $k$ time steps between consecutive predictions.

To overcome this challenge, we construct \emph{virtual predictions} to fill in between two consecutive actual predictions through an online metric matching algorithm that satisfies two properties, \emph{adherence} and \emph{strong competitiveness} (see \Cref{def:adhere,def:strongcomp}). Intuitively, these properties quantify how well the algorithm maintains an intermediate matching throughout its execution. 
While such intermediate guarantees are not required in the classical setting---since the total number of requests is known upfront---many existing algorithms already satisfy these properties (see \Cref{tab:alg:apply}). 

We leverage these algorithms $\calA$ within our meta-algorithm to generate virtual predictions as follows.
Denoting by $P$ the previous actual prediction, we treat $P$ as accurate, and we try to select additional servers from those unmatched in $P$ to include in the virtual predictions.
To this end, we define an auxiliary sub-instance in which the servers are those unmatched in $P$ and the requests are those arriving between $P$ and its subsequent actual prediction.
We then run $\calA$ on this sub-instance. 
At each time step, the servers matched by $\calA$, together with $P$, form the virtual prediction for that time step. 
To guarantee a good quality of the constructed virtual predictions, it is desirable for $\calA$ to maintain a high-quality matching throughout its execution, which is measured by adherence and strong competitiveness.

Next, we establish impossibility results by proving lower bounds on the cost incurred by any algorithm that uses predictions parsimoniously, even when the predictions are guaranteed to be accurate. These lower bounds are obtained by generalizing the standard hard instance on a star metric. It is noteworthy that our results show that the guarantee of \Cref{thm:intro:det} is essentially best-possible for deterministic algorithms up to a constant factor when the prediction is perfect, even in the more permissive setting where only the total number of predictions is bounded.

Finally, we empirically evaluate our theoretical results through experiments on instances that are either artificially synthesized or derived from real-world data.
Our experimental results are broadly consistent with our expectations, demonstrating the practical effectiveness of our approach.

Further related work is discussed in \Cref{app:relwork}.

\section{Preliminaries} \label{sec:prelim}
\paragraph{Online Metric Matching}
In the \emph{online metric matching} problem, we are given $n$ servers $S$ and $n$ requests $R$ contained in a metric $(V, d)$, i.e., for any $u, v, w \in V$, $d(u,u)=0$, $d(u,v)=d(v,u)$, and  $d(u, v) \leq d(u, w) + d(w, v)$.
Two distinct servers or requests can be placed at the same location in the metric; one can regard $S$ and $R$ as multisets of $V$.
The servers are known from the beginning whilst the requests are given one at a round in an online manner.
When a request $r \in R$ arrives, we need to assign $r$ to a currently unmatched server $s \in S$; this decision is irrevocable, i.e., $s$ is not allowed to be matched with any of subsequent requests, and $r$ cannot be reassigned to any other server in later rounds.
The objective is to find a perfect matching~$M$ between $S$ and $R$ at minimum total cost denoted by $\cost(M) := \sum_{(s, r) \in M} d(s, r)$.
For an online metric matching algorithm $\cA$, we also write $\cost(\cA)$ to denote the cost of matching output by $\cA$.

\begin{definition}
    For some $\rho : \bbN \to \bbR_+$, we say an algorithm $\cA$ is $\rho$-\emph{competitive} if $\cost(\cA) \leq \rho(n) \cdot \cost (\OPT)$ for any instance, where $\OPT$ denotes a minimum-cost perfect matching between $S$ and $R$ in the metric with hindsight.
\end{definition}

Throughout this paper, we assume $\cost(\OPT) \geq 1$.
This assumption is without loss of generality by scaling the metric with an appropriate factor; if $\cost(\OPT) = 0$, one can easily construct an optimal perfect matching even in the online setting.

\paragraph{Learning Augmentation with Action Predictions}
For the learning-augmented setting of online metric matching, \citet{antoniadis2023online} introduce the \emph{action prediction} defined as a server set $P_t \subseteq S$ of size $|P_t| = t$ for every round $t \in \{1, \ldots, n\}$.

The definition of prediction error requires a notion of \emph{distance} between two ``configurations'' which is formally defined as follows: for any multisets $A$ and $B$ of $V$ with $|A| = |B|$, we denote by $\dist(A, B)$ the minimum cost of a perfect matching between $A$ and $B$ in the metric $(V, d)$.
Here are some useful properties of this distance, whose proof is deferred to \Cref{app:dproofs}.
\begin{fact} \label{fact:prelim}
We have the following:
    \begin{enumerate}
        \item For multisets $A, B, C \subseteq V$ with $|A| = |B|$, $\dist(A, B) = \dist(A \uplus C, B \uplus C)$.
        \item For multisets $A, B, C \subseteq V$ with $|A| = |B| = |C|$, $\dist(A, B) \leq \dist(A, C) + \dist(C, B)$.
    \end{enumerate}
\end{fact}

We now define the error of action prediction.
Fix an optimal perfect matching $\OPT$ with hindsight.
For each $t \in \{1, \ldots, n\}$, let $r_t \in R$ denote the request that arrives in round~$t$, $R_t := \{r_1, \ldots, r_t\}$ be the first $t$ requests in their arrival order, and $O_t \subseteq S$ denote the set of servers matched by $R_t$ in $\OPT$.
The prediction error $\eta_t$ in round $t$ is then defined to be $\eta_t := \dist(P_t, O_t)$.

\paragraph{Parsimonious Learning Augmentation}
For the parsimonious learning-augmented setting \cite{im2022parsimonious,sadek2024algorithms}, it is assumed that there exists an oracle providing a prediction to the algorithm in any round on request from the algorithm.
We consider the following two parsimonious regimes studied by \citet{sadek2024algorithms}:
\begin{itemize}
    \item \emph{Well-separated queries to the oracle}: The algorithm receives a prediction once every $k$ time steps, for some \emph{separation parameter} $k$. Without loss of generality, the predictions are available to the algorithm in rounds $ik$ for every $i\in \{1, \ldots, \lfloor \nicefrac{n}{k}\rfloor \}$.
    \item \emph{Bounded number of predictions}: The algorithm is allowed to make a query to the oracle at any time under one hard constraint that the total number of queries must not exceed some budget $B \in \bbN$.
\end{itemize}
Note that the first regime is a special case of the second regime with budget $B = \lceil \nicefrac{n}{k} \rceil -1$.

The error of the oracle is defined to be the sum of errors of queried predictions, i.e., $\eta(Q) := \sum_{t \in Q} \eta_t$, where $Q \subseteq \{1, \ldots, n\}$ denotes the rounds in which the algorithm makes queries to the oracle.

\paragraph{Combination Algorithm}
In what follows, we focus on presenting a learning-augmented algorithm that performs well when the prediction is good, but is not robust when the prediction is bad, because there is a way to make any algorithm robust at a constant factor in cost \cite{fiat1994competitive,antoniadis2023online}.
For the sake of completeness, we provide the formal proof in \Cref{app:combalg}.
\begin{theorem} \label{thm:prelim:combalgrand}
    For two (possibly randomized) algorithms $\cA$ and $\cB$ for online metric matching, there exists a combination algorithm whose output incurs in expectation at most $9 \cdot \min \{ \E[\cost(\cA)], \E[\cost(\cB)] \}$.
\end{theorem}

\paragraph{Adherence and Strong Competitiveness}
Let us further define a couple of algorithmic properties that will be helpful in analysis of our parsimonious learning-augmented algorithm.
To define these properties, we fix an algorithm $\cA$ and an instance $(V, d, S, R)$. For $t \in \{1, \ldots, n\}$, let $\cM_t$ denote the set of all matchings where exactly $R_t$ is matched (i.e., all maximum matchings in the complete bipartite graph induced by $S \cup R_t$).
Moreover, let $\cA_t$ denote the matching maintained by $\cA$ at the end of round $t$, and let $S_t$ be the set of servers matched by $\cA_t$.
\begin{definition}[Adherence] \label{def:adhere}
    For some $\gamma \geq 1$, we say an algorithm $\cA$ is \emph{$\gamma$-adherent} if, for every $t \in \{1, \ldots, n\}$,
    $
        \textstyle \dist(S_t, R_t) \leq \gamma \cdot \min_{M \in \cM_t} \{\cost(M)\}.
    $
\end{definition}
\begin{definition}[Strong competitiveness] \label{def:strongcomp}
    For some $\rho : \bbN \to \bbR$, we say an algorithm $\cA$ is \emph{strongly $\rho$-competitive} if, for every $t \in \{1, \ldots, n\}$,
    $
        \textstyle \E[\cost(\cA_t)] \leq \rho(t) \cdot \min_{M \in \cM_t} \{\cost(M)\}.
    $    
\end{definition}
We remark that $O(1)$-adherence and strong competitiveness appear in many existing prediction-free algorithms.
We provide in \Cref{tab:alg:apply} a few such algorithms that we use to derive our parsimonious learning-augmented algorithms.
\section{Algorithms} \label{sec:alg}
In this section, we present our parsimonious learning-augmented algorithms for online metric matching.
They are built upon a meta-algorithm which is essentially the Follow-the-Prediction (FtP) algorithm of \citet{antoniadis2023online} that queries the prediction to the oracle in every $k$ round for some $k \in \bbN$ while generating \emph{virtual predictions} for the other rounds using an adherent, strongly competitive algorithm.
The main technical theorem about the meta-algorithm is as follows:
\begin{theorem} \label{thm:alg:main}
    Given a $\gamma$-adherent, strongly $\rho$-competitive algorithm and $k \in \bbN$, there is an algorithm that outputs a perfect matching of (expected) cost at most
    \[(1 + \gamma + \rho(k-1)) \, \cost (\OPT) + (2 + \gamma + \rho(k-1)) \, \eta(Q), \]
    where $Q = \{k, 2k, \ldots, \floor{\nicefrac{n}{k}} \cdot k \}$ denotes the rounds in which the algorithm makes queries to the oracle.
\end{theorem}
Note that the meta-algorithm works for the regime of well-separated queries to the oracle with separation parameter $k$, and hence, it also works for the more general regime of bounded number of predictions with budget $B = \floor{\nicefrac{n}{k}}$.

We begin by reviewing the FtP algorithm in \Cref{sec:alg:ftp}, and then describe the construction of (virtual) predictions and prove \Cref{thm:alg:main} in \Cref{sec:alg:query}.
Finally, in \Cref{sec:alg:apply}, we leverage existing algorithms within our meta-algorithm.

\subsection{Follow-the-Prediction Algorithm} \label{sec:alg:ftp}
The FtP algorithm is described in \Cref{alg:ftp}.
Note that this algorithm requires a prediction $P_t$ in every round $t \in \{1, \ldots, n\}$.

\begin{algorithm}
\caption{Follow-the-Prediction} \label{alg:ftp}
\begin{algorithmic}[1]
    \STATE $P_0 \gets \emptyset$, $S_0 \gets \emptyset$
    \FOR{each round $t \in \{1, \ldots, n\}$}
        \STATE $P_t \gets$ the prediction in round $t$
        \STATE $\mu_1 \gets$ a min-cost perfect matching between \\\quad$P_t$ and $P_{t-1} \cup \{r_t\}$ where every server in \\\quad$P_t\cap (P_{t-1}\cup\{r_t\})$ is matched to itself
        \STATE $p_t \gets$ the counterpart of $r_t$ in $\mu_1$
        \STATE $\mu_2 \gets$ a min-cost perfect matching between \\\quad $S \setminus P_{t-1}$ and $S \setminus S_{t-1}$ where every server \\\quad in $S \setminus (P_{t-1} \cup S_{t-1})$ is matched to itself
        \STATE $s_t \gets$ the counterpart of $p_t$ in $\mu_2$
        \STATE Match $r_t$ with $s_t$
        \STATE $S_t \gets S_{t-1} \cup \{s_t\}$
    \ENDFOR
\end{algorithmic}
\end{algorithm}

The FtP algorithm is guaranteed to output a feasible perfect matching between $S$ and $R$ because $p_t \in P_t \setminus P_{t - 1}$ due to $\mu_1$ and $s_t \in S \setminus S_{t-1}$ due to $\mu_2$ in every round $t$.

Following is the key lemma for competitiveness of the FtP algorithm.
We include its proof in \Cref{app:dproofs} for completeness' sake.
\begin{lemma}[\cite{antoniadis2023online}] \label{lem:alg:ftp}
    The total cost incurred by the FtP algorithm is
    $
        \textstyle \sum_{t = 1}^n d(s_t, r_t) \leq \sum_{t = 1}^n \dist (P_t, P_{t-1} \cup \{r_t\}).
    $
\end{lemma}

\newcommand{\ohat}{\widehat{o}}
\newcommand{\rhat}{\widehat{r}}
\newcommand{\shat}{\widehat{s}}
\newcommand{\Ohat}{\widehat{O}}
\newcommand{\Phat}{\widehat{P}}
\newcommand{\Rhat}{\widehat{R}}
\newcommand{\Shat}{\widehat{S}}
\newcommand{\OPThat}{\widehat{\OPT}}

\newcommand{\elli}[1]{\ell^{(#1)}}
\newcommand{\cAi}[1]{\cA^{(#1)}}
\newcommand{\ohati}[1]{\ohat^{(#1)}}
\newcommand{\rhati}[1]{\rhat^{(#1)}}
\newcommand{\shati}[1]{\shat^{(#1)}}
\newcommand{\Ohati}[1]{\Ohat^{(#1)}}
\newcommand{\Phati}[1]{\Phat^{(#1)}}
\newcommand{\Rhati}[1]{\Rhat^{(#1)}}
\newcommand{\Shati}[1]{\Shat^{(#1)}}
\newcommand{\OPThati}[1]{\OPThat^{(#1)}}

\subsection{Construction of Virtual Predictions} \label{sec:alg:query}
In this subsection, we present the constructions of virtual predictions and prove \Cref{thm:alg:main}.
Recall that the FtP algorithm requires a prediction in every round while we make queries to the oracle only for every $k$ rounds.
We therefore need to construct virtual predictions for the remaining rounds.

We achieve this goal by a $\gamma$-adherent, strongly $\rho$-competitive algorithm $\cA$ for some $\gamma \geq 1$ and $\rho : \bbN \to \bbR_+$.
Intuitively speaking, for the rounds after querying $\Phat$ to the oracle until making the next query (or termination), we execute $\cA$ as a subroutine given an auxiliary instance where the servers are $S \setminus \Phat$ and the requests are those residual after querying $\Phat$.
We then serve $\Phat \cup \Shat$ to the FtP algorithm as the prediction in each round, where $\Shat \subseteq S \setminus \Phat$ denotes the servers matched by $\cA$ in the auxiliary instance at the end of each round.
A detailed pseudocode is provided in \Cref{alg:query}.

\begin{algorithm}
\caption{Construction of (virtual) predictions} \label{alg:query}
\begin{algorithmic}[1]
    \STATE $\Phat \gets \emptyset$
    \STATE Initialize a new instance for $\cA$ with servers $S$
    \FOR{each round $t \in \{1, \ldots, n\}$}
        \IF{$t$ is a multiple of $k$}
            \STATE Query $P_t$ to the oracle
            \STATE $\Phat \gets P_t$
            \STATE Initialize a new instance for $\cA$ with servers $S \setminus \Phat$
        \ELSE
            \STATE Feed $r_t \in R$ to $\cA$
            \STATE Let $\Shat \subseteq S \setminus \Phat$ be the servers matched by $\cA$ \\\quad until now
            \STATE $P_t \gets \Phat \cup \Shat$
        \ENDIF
    \ENDFOR
\end{algorithmic}
\end{algorithm}

Our meta-algorithm is then the FtP algorithm (\Cref{alg:ftp}) with this construction of predictions (\Cref{alg:query}).
It is not hard to see that the constructed predictions are feasible, i.e., $P_t \subseteq S$ and $|P_t| = t$ for every $t \in \{1, \ldots, n\}$, implying that our algorithm is well-defined.

It remains to show the competitiveness of our algorithm to prove \Cref{thm:alg:main}.
To this end, let us define more notation.
Recall that we fix an optimal perfect matching $\OPT$ with hindsight.
For every $t \in \{1, \ldots, n\}$, recall also that $r_t \in R$ is the request arriving in round $t$, and let $o_t$ denote the server matched with $r_t$ in $\OPT$.

Note that we construct virtual predictions using $\cA$ given an auxiliary instance between two consecutive rounds making queries to the oracle.
We call the maximal subsequence of rounds for each auxiliary instance a \emph{phase}.
More precisely, for $i \in \{0, \ldots, \floor{\nicefrac{n}{k}}\}$, phase $i$ consists of the rounds $\{ ik+1, \ldots, ik + (k-1) \}$ (or $\{ik + 1, \ldots, n\}$ for the last phase).
Let $\elli{i}$ denote the length of phase $i$, i.e., $\elli{i} = n - ik$ if $i = \floor{\nicefrac{n}{k}}$ and $\elli{i} = k-1$ otherwise.

For $i \in \{0, \ldots, \floor{\nicefrac{n}{k}}\}$, let $\Rhati{i} := \{ r_{ik+1}, \ldots, r_{ik+\elli{i}} \}$ denote the requests fed to $\cA$ during phase $i$, and let $\OPThati{i}$ denote a minimum-cost ``maximum'' matching between $S \setminus P_{ik}$ and $\Rhati{i}$ where $\Rhati{i}$ is all matched.
We first show that the cost of $\OPThati{i}$---an ``optimal'' matching in the auxiliary instance of phase $i$---is bounded.
\begin{lemma} \label{lem:alg:aux}
    For each $i \in \{0, \ldots, \floor{\nicefrac{n}{k}}\}$, 
    \[
        \textstyle \cost (\OPThati{i}) \leq \sum_{j = 1}^{\elli{i}} d(o_{ik+j}, r_{ik+j}) + \eta_{ik}.
    \]
\end{lemma}
\begin{proof}
    As the phase is clear from context, we omit the phase in the notation and simply write $\Rhat$, $\ell$, and $\OPThat$. 
    Let $\Phat := P_{ik}$ denote the prediction queried the last to the oracle until then.
    For every $j \in \{1, \ldots, \ell\}$, we further write $\rhat_j := r_{ik + j}$ for notational simplicity.
    Let $\ohat_j \in S \setminus \Phat$ denote the server matched with $\rhat_j$ in $\OPThat$, and let $\Ohat \subseteq S \setminus \Phat$ denote the all servers matched with $\Rhat$ in $\OPThat$; we hence have
    \[
        \cost(\OPThat) = \dist(\Ohat, \Rhat) = \dist(\Ohat \cup \Phat, \Rhat \cup \Phat),
    \]
    where the last equality is due to \Cref{fact:prelim}.

    We claim that, for any $T \subseteq S$ of size $|T| = |\Ohat \cup \Phat| = ik+\ell$, we have
    \begin{equation} \label{ineq:alg:aux1}
        \dist(\Ohat \cup \Phat, \Rhat \cup \Phat) \leq \dist(T, \Rhat \cup \Phat).
    \end{equation}
    If $T \supseteq \Phat$, the claim immediately follows from the optimality of $\Ohat$ and the fact that $\dist(T, \Rhat \cup \Phat) = \dist(T \setminus \Phat, \Rhat)$.
    
    We now consider the case where $T \not\supseteq \Phat$.
    In this case, we argue that there always exists $T' \supseteq \Phat$ with $|T'| = |T|$ such that $\dist(T', \Rhat \cup \Phat) \leq \dist(T, \Rhat \cup \Phat)$.
    For any $u \in \Phat \setminus T$, let $v \in T \setminus \Phat$ be the counterpart of $u$ in a minimum-cost perfect matching $\mu$ between $T\setminus \Phat$ and $(\Rhat \cup \Phat)\setminus T$.
    Now consider $T' := (T \setminus \{v\}) \cup \{u\}$.
    Since substituting matching edge $\{u, v\}$ with edge $\{u, u\}$ in $\mu$ yields a feasible perfect matching between $T'$ and $\Rhat \cup \Phat$, we have
    \begin{align*}
        \dist(T', \Rhat \cup \Phat) 
        &
        \leq \dist(T, \Rhat \cup \Phat) - d(u, v) + d(u, u)
        \\&
        \leq \dist(T, \Rhat \cup \Phat)
    \end{align*}
    with $|T' \cap \Phat| > |T \cap \Phat|$.
    Repeating this substitution completes the proof of this case.

    Recall that $O_t \subseteq S$ denotes the servers matched with the first $t$ requests $R_t$ in $\OPT$ for every $t \in \{1, \ldots, n\}$. We can now derive
    \begin{align*}
        &
        \cost(\OPThat)
        \leq \dist(O_{ik + \ell}, \Rhat \cup \Phat)
        \\&
        \qquad \leq \dist(O_{ik + \ell}, O_{ik} \cup \Rhat) + \dist(O_{ik} \cup \Rhat, \Rhat \cup \Phat)
        \\&
        \qquad \textstyle = \sum_{j=1}^{\ell} d(o_{ik + j}, r_{ik+j}) + \eta_{ik},
    \end{align*}
    where the first inequality comes from \Cref{ineq:alg:aux1}, and the second from \Cref{fact:prelim}.
\end{proof}

Notice that, due to \Cref{lem:alg:ftp}, the cost of the matching output by the FtP algorithm is bounded by $\sum_{t = 1}^n \dist(P_t, P_{t-1} \cup \{r_t\})$.
We split this expression into phases and separately bound the contribution of each phase as follows.

\begin{lemma}\label{lem:alg:phase}
    For each $i \in \{0, \ldots, \floor{\nicefrac{n}{k}}\}$,
    \begin{align*}
        & \textstyle \E \Big[ \sum_{j = 1}^{\elli{i}} \dist(P_{ik + j}, P_{ik + j-1} \cup \{r_{ik + j}\}) \Big] \\
        & \qquad \textstyle \leq \rho(\elli{i}) \cdot \Big( \sum_{j = 1}^{\elli{i}} d(o_{ik+j}, r_{ik+j}) + \eta_{ik} \Big).
    \end{align*}
\end{lemma}
\begin{proof}
    For each $j \in \{1, \ldots, \elli{i}\}$, let $\shat_j \in S \setminus P_{ik}$ denote the server matched with $r_{ik + j}$ by $\cA$ in phase $i$.
    Observe that
    \[
        \dist(P_{ik + j}, P_{ik + j - 1} \cup \{r_{ik + j}\}) = d(\shat_j, r_{ik + j}),
    \]
    implying that the left-hand side of the inequality in the lemma statement is precisely the total cost incurred by $\cA$ given the auxiliary instance of phase $i$ where $\{r_{ik+1}, \ldots, r_{ik+\elli{i}}\}$ are fed as requests.
    Since $\cA$ is strongly $\rho$-competitive, this cost is bounded from above by $\rho(\elli{i}) \cdot \cost (\OPThati{i})$ in expectation.
    The proof of this lemma then follows from \Cref{lem:alg:aux}.
\end{proof}

\begin{lemma} \label{lem:alg:query}
    For each $i \in \{1, \ldots, \floor{\nicefrac{n}{k}}\}$,
    \begin{align*}
        & \dist(P_{ik}, P_{ik-1} \cup \{r_{ik}\}) \\
        & \qquad \textstyle \leq (1 + \gamma) \sum_{j=1}^{k} d(o_{(i-1)k + j}, r_{(i-1)k+j}) \\
        & \qquad \phantom{\leq} + (1 + \gamma) \eta_{(i-1)k} + \eta_{ik}.
    \end{align*}
\end{lemma}
\begin{proof}
    Let $\Rhat := \Rhati{i-1}$ be the requests in phase $(i-1)$, and let $\Shat \subseteq S \setminus P_{(i-1)k}$ denote the servers matched with $\Rhat$ by $\cA$ in phase $(i-1)$.
    We can deduce the following:
    \begin{align*}
        & 
        \dist(P_{ik}, P_{ik-1} \cup \{r_{ik}\}) 
        \\&
        \quad \leq \dist(P_{ik}, O_{ik})
        + \dist(O_{ik}, O_{(i-1)k} \cup \Rhat \cup \{r_{ik}\})
        \\& \quad \phantom{\leq} 
        + \dist(O_{(i-1)k} \cup \Rhat \cup \{r_{ik}\}, O_{(i-1)k} \cup \Shat \cup \{r_{ik}\})
        \\&
        \quad \phantom{\leq} + \dist(O_{(i-1)k} \cup \Shat \cup \{r_{ik}\}, P_{ik-1} \cup \{r_{ik}\} )
        \\&
        \quad = \eta_{ik}
        + \sum_{j = 1}^{k} d(o_{(i-1)k+j}, r_{(i-1)k+j})
        \\&\quad \phantom{=}
        + \dist(\Rhat, \Shat)
        + \eta_{(i-1)k}
        \\&
        \quad \textstyle \leq \eta_{ik}
        + \sum_{j = 1}^{k} d(o_{(i-1)k+j}, r_{(i-1)k+j})
        \\&\quad  \phantom{=}
        + \gamma \cdot \cost(\OPThati{i-1})
        + \eta_{(i-1)k}
        \\& \quad 
        \leq \textstyle (1 + \gamma) \sum_{j=1}^{k} d(o_{(i-1)k + j}, r_{(i-1)k+j}) 
        \\&\quad 
        \phantom{\leq} + (1 + \gamma) \eta_{(i-1)k} + \eta_{ik},
    \end{align*}
    where the first inequality and the equality come from \Cref{fact:prelim}, the second inequality from $\gamma$-adherence of $\cA$, and the last inequality from \Cref{lem:alg:aux}.
\end{proof}

We are now ready to prove the main theorem.

\begin{proof}[Proof of \Cref{thm:alg:main}]
    Immediate from \Cref{lem:alg:ftp,lem:alg:phase,lem:alg:query}.
\end{proof}
\subsection{Applications} \label{sec:alg:apply}

Our meta-algorithm takes as a subroutine a $\gamma$-adherent, strongly $\rho$-competitive algorithm $\cA$, and many existing algorithms in the literature satisfy these properties.
See \Cref{tab:alg:apply} for a summary on a few algorithms that we use in our learning-augmented algorithm; we prove the adherence and strong competitiveness of these algorithms in \Cref{app:compalgs}.

\begin{table}
    \begin{center}
    \begin{small}
    \begin{sc}
    \caption{$\gamma$-adherence and strong $\rho$-competitiveness of online metric matching algorithms: \textsc{KP93} \cite{kalyanasundaram1993online}, \textsc{KMV94} \cite{khuller1994line}, \textsc{NR17} \cite{nayyar2017input}, \textsc{Rag18} \cite{raghvendra2018optimal}, and \textsc{BBGN14} \cite{bansal2014randomized}.} \label{tab:alg:apply}
    \begin{tabular}{ccccc}
        \toprule
        Algorithm & Metric & Det? & $\gamma$ & $\rho(x)$ \\ \midrule
        KP93, KMV94
        & General & Yes & $1$ & $2x-1$ \\
        NR17
        & General & Yes & $3$ & $O(\mu_\bbM(S) \log^2 x)$ \\
        Rag18
        & Line & Yes & $3$ & $O(\log x)$ \\
        BBGN14
        & 2-HST & No & $1$ & $O(\log x)$ \\ \bottomrule
    \end{tabular}
    \end{sc}
    \end{small}
    \end{center}
\end{table}

Using the classical deterministic algorithm \cite{kalyanasundaram1993online,khuller1994line} for $\cA$ gives the following deterministic algorithm.

\begin{corollary} [cf.~\Cref{thm:intro:det}] \label{thm:alg:detgen}
For a general metric and $k \in \bbN$, there exists a deterministic learning-augmented algorithm that outputs a perfect matching of cost at most $(2k-1) \, \cost (\OPT) + 2k \, \eta (Q)$ with making $\floor{\nicefrac{n}{k}}$ queries well-separated with parameter $k$.
\end{corollary}

The algorithm presented in \citet{nayyar2017input} (also in \citet{raghvendra2018optimal}) gives the following corollaries:

\begin{corollary} \label{thm:alg:detgen2}
For a general metric and $k \in \bbN$, there exists a deterministic learning-augmented algorithm that outputs a perfect matching of cost at most $O(\mu_\bbM(S) \log^2 k) \cdot (\cost (\OPT) + \eta (Q))$ with making $\floor{\nicefrac{n}{k}}$ queries well-separated with parameter $k$, where $\mu_\bbM(S)$ denotes the maximum ratio of the traveling salesman tour
and the diameter of any subset of S.
\end{corollary}

\begin{corollary} \label{thm:alg:detline}
For a line metric and $k \in \bbN$, there exists a deterministic learning-augmented algorithm that outputs a perfect matching of cost at most $O(\log k) \cdot (\cost (\OPT) + \eta (Q))$ with making $\floor{\nicefrac{n}{k}}$ queries well-separated with parameter $k$.
\end{corollary}

Due to \citet{bansal2014randomized}, we have the following algorithm for 2-HSTs (see \Cref{def:app:hst}).

\begin{corollary} \label{thm:alg:randhst}
For a metric induced by a 2-HST and $k \in \bbN$, there exists a randomized learning-augmented algorithm that outputs a perfect matching of expected cost at most $O(\log k) \cdot (\cost (\OPT) + \eta (Q))$ with deterministically making $\floor{\nicefrac{n}{k}}$ queries well-separated with parameter $k$.
\end{corollary}

As any metric can be embedded into a 2-HST with $O(\log n)$ expected distortion (\cite{fakcharoenphol2004tight}, see \Cref{thm:app:hst}), we further derive the following result.

\begin{corollary} [cf.~\Cref{thm:intro:rand}] \label{thm:alg:randgen}
For a general metric and $k \in \bbN$, there exists a randomized learning-augmented algorithm that outputs a perfect matching of expected cost at most $O(\log n \log k) \cdot (\cost (\OPT) + \eta (Q))$ with deterministically making $\floor{\nicefrac{n}{k}}$ queries well-separated with parameter $k$.
\end{corollary}
\section{Lower Bounds} \label{sec:hard}
We establish lower bounds on competitiveness of parsimonious learning-augmented algorithm for online bipartite matching.
Due to space constraints, we defer the whole analysis to \Cref{app:hard}.
We call an oracle that always gives an accurate prediction an \emph{omnipotent oracle}.

For deterministic algorithms, our lower bounds complement that the algorithm guaranteed in \Cref{thm:alg:detgen} is best-possible (up to a small constant factor) when the prediction is accurate in the regimes of well-separated queries to the oracle (and bounded number of predictions). 
\begin{theorem} \label{thm:hard:det}
    No deterministic algorithms making at most $B$ queries to the omnipotent oracle that always gives an accurate prediction can achieve a competitive ratio strictly better than $\frac{2n}{B+1} - 1$ for any metric space.
\end{theorem}
\begin{theorem} \label{thm:hard:detwsq}
    No deterministic algorithm making well-separated queries to the omnipotent oracle with separation parameter $k$ can achieve a competitive ratio strictly better than $2k-1$ for any metric space.
\end{theorem}

For randomized algorithms, we present lower bounds on the class of algorithms making deterministic queries for the two regimes, respectively.
\begin{theorem} \label{thm:hard:randbnp}
    No randomized algorithm that deterministically makes at most $B$ queries to the omnipotent oracle can achieve a competitive ratio of $1+o(\frac{\log (\nicefrac{n}{B}) }{B})$ for any metric space.
\end{theorem}

\begin{theorem} \label{thm:hard:randwsq}
    No randomized algorithm making well-separated queries to the omnipotent oracle with separation parameter $k$ can achieve a competitive ratio of $o(\log k)$ for any metric space.
\end{theorem}
\newcommand{\greedy}{\textsf{Greedy}}
\newcommand{\comp}{\textsf{Comp}}
\newcommand{\ours}{\textsf{Ours}}
\newcommand{\combwgreedy}{\textsf{Comb-Greedy}}
\newcommand{\combwcomp}{\textsf{Comb-Comp}}

\newcommand{\lineinst}{\textsf{Line}}
\newcommand{\planeinst}{\textsf{Plane}}
\newcommand{\taxiinst}{\textsf{Taxi}}

\section{Experiments} \label{sec:exp}
This section presents an empirical evaluation of our algorithms.
We begin by describing the experimental setup in \Cref{sec:exp:setup}, followed by a discussion of the experimental results in \Cref{sec:exp:res}.

\begin{figure*}[t]
    \centering
    \includegraphics[width=\linewidth]{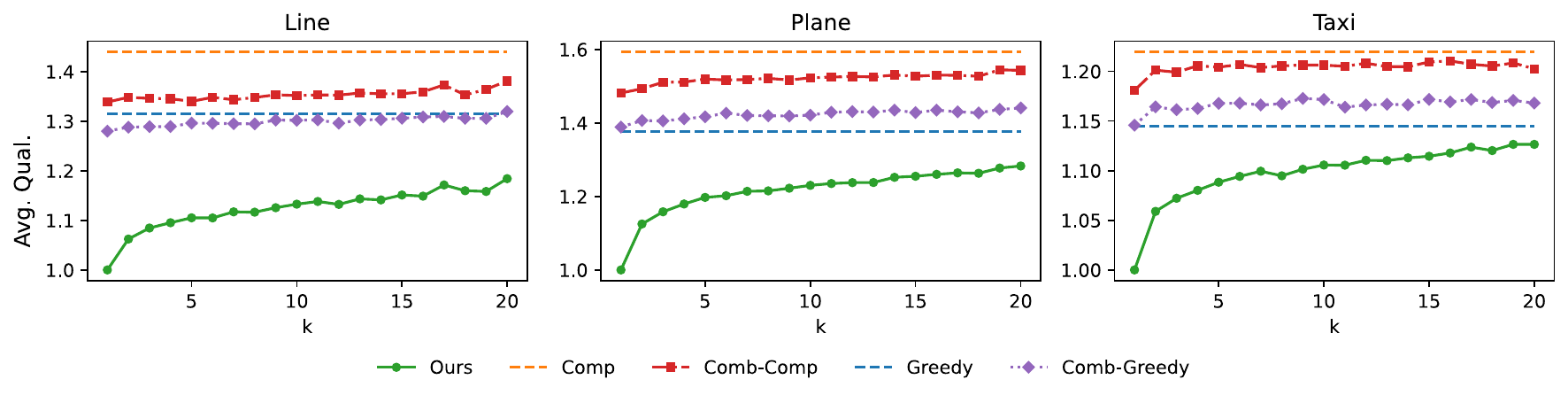}
    \caption{
    \textbf{Results of the first type of experiments.} 
    The plots for \lineinst{}, \planeinst{}, and \taxiinst{} are shown on the left, center, and right, respectively. 
    The $x$-axis corresponds to the separation parameter $k$, and the $y$-axis represents the average solution quality, measured by the ratio of the algorithm's output cost to the optimal cost with hindsight.
    }
    \label{fig:exp:perfect}
\end{figure*}
\begin{figure*}[t]
    \centering
    \includegraphics[width=\linewidth]{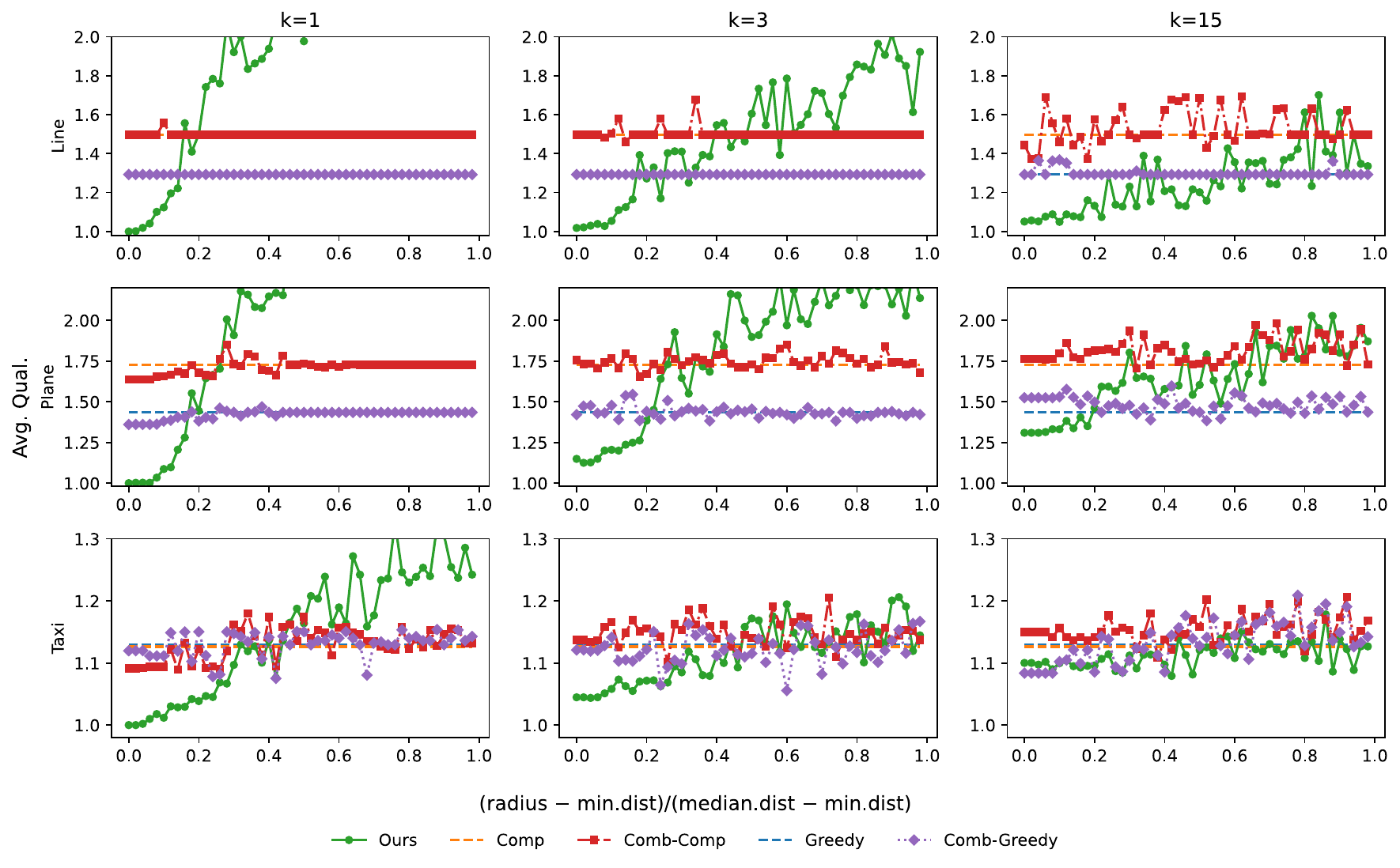}
    \caption{
    \textbf{Representative results of the second type of experiments.}
    The rows correspond to the instance classes, while the columns correspond to different values of the separation parameter $k$.
    The $x$-axis displays the normalized prediction noise, represented by $\frac{r - d_{\min}}{d_{\mathrm{med}} - d_{\min}}$, where $r$ is the noise radius, $d_{\min}$ is the minimum non-zero distance, and $d_{\mathrm{med}}$ is the median of all non-zero distance values. 
    The $y$-axis represents the average solution quality, measured by the ratio of the algorithm's output cost to the optimal cost with hindsight.
    }
    \label{fig:exp:noise}
\end{figure*}

\subsection{Experimental Setup} \label{sec:exp:setup}
\paragraph{Algorithms}
We evaluate the following five algorithms in our experiments:
\begin{itemize}
    \item \ours{}, the algorithm described in \Cref{thm:alg:detgen};
    \item \comp{}, the deterministic $(2n-1)$-competitive algorithm due to \citet{kalyanasundaram1993online} and \citet{khuller1994line};
    \item \combwcomp{}, the combination of \ours{} and \comp{} as formalized in \Cref{thm:intro:det};
    \item \greedy{}, the simple greedy algorithm that matches every request to its nearest available server;
    \item \combwgreedy{}, the combination of \ours{} and \greedy{}.
\end{itemize}
We remark that, when setting $k = 1$, \ours{} is equivalent with the Follow-the-Prediction algorithm of \citet{antoniadis2023online}.

\paragraph{Instance Generation}
Each instance is defined on a metric space with 200 vertices, consisting of 100 servers and 100 requests.
We consider three classes of instances---\lineinst{}, \planeinst{}, and \taxiinst{}---which are constructed as follows.

To generate a \lineinst{} instance, we independently sample 200 numbers uniformly at random from the interval $[0,1]$, each representing the location of a vertex on the line.
The distance between two vertices is defined as the absolute difference between their corresponding numbers.
Among these 200 vertices, we sample 100 uniformly at random without replacement as servers, and sample another 100 uniformly at random with replacement as requests.

An instance of \planeinst{} is constructed similarly to \lineinst{}.
The only difference is that each vertex corresponds to a point sampled uniformly at random from the two-dimensional unit square $[0,1]^2$, and distances are measured using the Euclidean metric, with a modest rounding applied for numerical stability.

Inspired by \citet{feng2024two} and \citet{jin2022online}, \taxiinst{} denotes a class of instances constructed from the Chicago Taxi Trips (2013--2023) dataset~\cite{chicago_taxi_trips}.
Specifically, for a randomly sampled time point, the most recent 100 data entries preceding this time are treated as servers, while the earliest 100 entries following it are treated as requests arriving sequentially.
Distances are measured using the Manhattan metric.
In the interest of space, we defer the full details of \taxiinst{} instance construction to \Cref{app:taxigen}.

\paragraph{Noisy Prediction Generation}
Our construction of noisy predictions is parameterized by a noise radius $r \geq 0$.
For each round $t$, let $P_t^\star \subseteq S$ denote the accurate prediction at that round.
For every server $s \in P_t^\star$, we independently select a server $q_s$ uniformly at random from those within distance $r$ of $s$.
Let $Q_t := \{ q_s \mid s \in P_t^\star \}$ denote the resulting multiset; note that repetitions may occur, i.e., it is possible that $q_{s_1} = q_{s_2}$ for distinct $s_1, s_2 \in P_t^\star$.
To eliminate such repetitions, we compute a minimum-cost matching between $Q_t$ and $S$ in which every server of $Q_t$ is matched.
The prediction $P_t \subseteq S$ for round $t$ is then defined as the set of (distinct) servers in $S$ matched to $Q_t$.
Observe that 
$
\dist(P_t^\star, P_t) \le 2 t r,
$
since the matching $\{(q_s, s) \mid s \in P_t^\star\}$ is a feasible matching between $Q_t$ and $S$ and has total cost at most $t r$.

\paragraph{Evaluation Objectives}
We consider two types of experiments, each designed to evaluate different aspects of the algorithms.
The first type focuses on the average performance of the algorithms as a function of the number of actual predictions used.
To eliminate potential bias arising from prediction errors, we use only perfect predictions in this set of experiments.
Specifically, for each instance class, we compute the average solution quality—measured as the ratio between the algorithm’s cost and the optimal cost—over 100 independently generated instances, for separation parameters $k \in \{1,\ldots,20\}$.

The second type of experiments aims to evaluate performance under noisy predictions.
Since the same level of prediction noise may affect different instances in different ways, we fix an instance and compute the average solution quality over 100 independently generated predictions.
We consider 50 noise radii, uniformly spaced between the minimum non-zero distance and the median non-zero distance, and evaluate performance for separation parameters $k \in \{1,\ldots,20\}$.
This set of experiments is respectively conducted on 5 randomly sampled instances for each instance class.

\subsection{Experimental Results} \label{sec:exp:res}
The results of the first type of experiments are shown in \Cref{fig:exp:perfect}.
Overall, the observed trends align well with our expectations.
In particular, the performance of \ours{}, \combwcomp{}, and \combwgreedy{} degrades as the separation parameter $k$ increases.
Among them, the performance of \ours{} deteriorates more rapidly than that of \combwcomp{} and \combwgreedy{}, which is expected since combining with \comp{} or \greedy{} improves robustness.

On the other hand, \ours{} consistently outperforms the other algorithms.
We attribute this to its ability to fully exploit the available accurate predictions.
One notable observation is that \combwgreedy{} is outperformed by \greedy{} on instances of \planeinst{} and \taxiinst{}.
We discuss this behavior at the end of this subsection.

The results of the second type of experiments are summarized in \Cref{fig:exp:noise}.
Since these experiments produce 100 plots for each instance class, we present only a representative subset that clearly illustrates the general trends, due to space constraints.
The complete set of plots is provided in the Supplementary Material.

Overall, the experimental results again align well with our expectations.
For \ours{}, we observe that the algorithm is more sensitive to prediction noise when the separation parameter $k$ is small.
In particular, when $k = 1$, the algorithm achieves near-optimal performance under very low noise, but its performance deteriorates rapidly as the noise level increases.
Recall that \ours{} with $k = 1$ is equivalent to the FtP algorithm of \citet{antoniadis2023online}.
As $k$ increases, the slope of performance degradation becomes noticeably milder, although the performance under accurate predictions also worsens.

This behavior is expected: the more frequently the algorithm exploits predictions, the more strongly it is affected by prediction noise.
The combination algorithms exhibit similar qualitative trends, but are substantially more robust to noise due to the presence of the fallback algorithms \comp{} and \greedy{}.

We finally discuss the performance of the combination algorithm.
Across both types of experiments, the combination algorithm underperforms compared to its fallback algorithm in many cases. 
This behavior is likely explained by the overhead incurred when switching between the two input algorithms.
Consequently, if the performance gap between the input algorithms is relatively narrow, this switching overhead may become significant, potentially leading to worse overall performance for the combination algorithm.
\section{Conclusion}
In this paper, we extend the study of parsimonious prediction models in learning-augmented online optimization~\cite{im2022parsimonious,sadek2024algorithms} to the setting of online metric matching with action predictions~\cite{antoniadis2023online}.

Several directions remain open for future work.
An immediate direction is to develop a more refined analysis for randomized algorithms in the parsimonious setting.
In particular, it would be interesting to determine whether one can achieve performance guarantees that depend only on the separation parameter~$k$, rather than the $O(\log n \log k)$ factor in \Cref{thm:alg:randgen}, as well as lower bounds for algorithms that make randomized queries, rather than \Cref{thm:hard:randbnp}.

The combination algorithm is essential for robustifying our algorithmic framework. 
However, its current version incurs a multiplicative factor of 9 in the competitive ratio. 
Moreover, our experiments show that the combination algorithm often underperforms relative to its fallback algorithm. 
Designing a theoretically improved and/or practically effective combination algorithm is therefore an important direction for future work.

Another promising direction is to investigate alternative prediction models for online metric matching.
One plausible candidate is a variant of the \emph{metric error with outliers} model proposed by \citet{azar2022online}, despite the fact that the original model is known not to admit strong guarantees for online metric matching.

\section*{Acknowledgements}
This work was supported by the NCN grant no.\ 2020/39/B/ST6/01641.
This work was partly supported by Institute of Information \& communications Technology Planning \& Evaluation (IITP) grant funded by the Korea government (MSIT) (No.\ RS-2021-II212068, Artificial Intelligence Innovation Hub).
YS thanks Sungjin Im for introducing this problem and Hyung-Chan An for valuable discussions.
The authors are also grateful to Christian Coester \yrcite{coester2026} for acknowledging an inconsistency in \citet{antoniadis2023online}.
Part of this work was done while YS was at Yonsei University, South Korea.




\section*{Impact Statement}
This paper presents theoretical results on parsimonious learning-augmented algorithms for online metric matching, extending prior work on parsimonious learning-augmented online optimization.
Given the theoretical nature of this work, we do not anticipate direct negative societal impacts.


\bibliography{ref}
\bibliographystyle{icml2026}

\newpage
\appendix
\onecolumn

\section{Further Related Work} \label{app:relwork}

\paragraph{Online Metric Matching}
Online metric matching has been studied under a variety of metric spaces.
\citet{gupta2012online} presented a randomized $O(\log n)$-competitive algorithm for line and doubling metrics.
For deterministic algorithms, \citet{antoniadis2019competitive} gave the first sublinear-competitive algorithm for the line metric.
Subsequently, \citet{raghvendra2018optimal} developed a deterministic $O(\log n)$-competitive algorithm for line metrics, building on an instance-dependent $O(\mu_{\bbM}(S)\log^2 n)$-competitive algorithm due to \citet{nayyar2017input}, where $\mu_{\bbM}(S)$ denotes the maximum ratio between the traveling salesman tour and the diameter over all subsets of the server set~$S$.
On the lower-bound side, \citet{peserico2023matching} recently showed that any algorithm, including randomized ones, must have a competitive ratio of at least $\Omega(\sqrt{\log n})$ on the line metric, improving upon the previous lower bound of $9.001$ for deterministic algorithms established by \citet{fuchs2005online}.

The problem has also been studied under other arrival models than the adversarial model considered in this paper.
\citet{raghvendra2016robust} considered the random arrival model and obtained a best-possible $O(\log n)$-competitive algorithm, together with a matching lower bound.
In the known i.i.d.~model—where each request is drawn independently from a known distribution—\citet{gupta2019stochastic} presented an $O((\log\log\log n)^2)$-competitive algorithm for general metrics, which achieves a constant competitive ratio of $9$ for line and tree metrics.

\paragraph{Parsimonious Usage of Predictions}
Besides \citet{im2022parsimonious} and \citet{sadek2024algorithms}, a growing body of work has investigated how to reduce the amount of predictive information used in learning-augmented algorithms.
\citet{drygala2023online} studied the ski rental and Bahncard problems in settings where using predictions incurs an additional cost.
\citet{benomar2023advice} examined the tradeoff between algorithmic performance and a bounded number of predictions for problems such as ski rental, secretary, and non-clairvoyant scheduling, under a model in which predictions are accurate with known probabilities.
\citet{bampis2024parsimonious} showed that a logarithmic number of predictions suffices to improve the running time of PTASs for dense instances of several NP-hard problems.
\citet{antoniadis2023paging} introduced single-bit prediction models for paging and developed learning-augmented algorithms for both models, together with matching lower bounds.
Finally, we note that the parsimonious use of predictions in learning-augmented algorithms is closely related to the literature on \emph{online optimization with advice}, which studies the bit complexity of accurate advice required to achieve given performance guarantees; we refer interested readers to the survey by \citet{boyar2017online}.

\paragraph{Learning-Augmented Online Matching and Beyond}
Online bipartite matching has also been extensively studied in the learning-augmented setting.
\citet{mahdian2012online} initiated this line of work by developing a learning-augmented algorithm for the AdWords problem under the small-bids assumption; their result was later improved by \citet{choo2025learning}.
\citet{spaeh2023online} extended this framework to the Display Ads and generalized assignment problems.
\citet{jin2022online} investigated a two-stage model of unweighted/vertex-weighted/edge-weighted online bipartite matching.
Related work has also considered the random arrival model \cite{antoniadis2023secretary, choo2024online} and the stochastic graph setting \cite{aamand2022optimal}.
Beyond online bipartite matching, a substantial body of research has explored learning-augmented algorithms more broadly; interested readers are referred to the survey by \citet{mitzenmacher2022algorithms} and the website~\cite{alps}, which provides a comprehensive collection of publications in the area.
 
\newcommand{\calCbar}{\overline{\calC}}
\newcommand{\MjC}[2]{M^{(#1)}_{#2}}
\newcommand{\symdiff}{\mathbin{\triangle}}

\section{Combination Algorithm} \label{app:combalg}
In this appendix, we prove \Cref{thm:prelim:combalgrand}.

\paragraph{Algorithm Description}
We first introduce some notation:
For an online algorithm $\calC$ and any $t \in \{1, \ldots, n\}$, let $M^\calC_t$ be the intermediate matching that $\calC$ maintains at the end of round $t$.
Let $s^\calC_t \in S$ and $S^\calC_t \subseteq S$ denote the server matched by $\calC$ at round $t$ and the set of servers matched by $\calC$ so far, respectively.
Let $S^\calC_0 = \emptyset$.

We now describe the combination algorithm.
For every round $t \in \{1, \ldots, n\}$, let $S_t \subseteq S$ denote the servers matched by this combination algorithm at the end of round $t$; let $S_0 = \emptyset$.
Recall that we are given two online algorithms $\calA$ and $\calB$. 
We divide the entire rounds into \emph{phases}, where we try to follow $\calA$ in every odd phase, and $\calB$ in every even phase.

At the beginning of a round $t \in \{1, \ldots, n\}$, let us assume that we are in phase $i$ for some odd $i$.
We first determine if the cost incurred by $\calA$ is greater than $2^i$ or not.
If it is the case (i.e., $\cost(M^\calA_{t}) > 2^i$), we move to the next phase without processing the request of this round.
Here we remark that no rounds can be processed in a phase.

Otherwise, if $\cost(M^\calA_{t-1}) \leq 2^i$, let $s_t$ denote the counterpart of $s^\calA_t$ in a minimum-cost perfect matching $\mu_t$ 
between $S \setminus S^\calA_{t-1}$ and $S \setminus S_{t-1}$, 
where every $s \in S \setminus (S_{t-1} \cup S^\calA_{t-1})$ is matched with itself.
We then match $r_t$ with $s_t$. 
For an even phase, we replace every $\calA$ with $\calB$ in the description.
This is the end of the description of the combination algorithm.

The pseudocode of the combination algorithm is provided in \Cref{alg:app:combalgdet}.

\begin{algorithm}
  \caption{Combination Algorithm}
  \label{alg:app:combalgdet}
  \begin{algorithmic}
    \STATE $i \gets 1$, $t \gets 1$, $S_0 \gets \emptyset$
    \WHILE{$t \leq n$}
        \STATE $\calC \gets \begin{cases}\calA, & \text{if $i$ is odd} \\ \calB, & \text{if $i$ is even} \end{cases}$
        \IF{$\cost(M^\calC_{t}) > 2^i$}
            \STATE $i \gets i+1$
        \ELSE
            \STATE $s^\calC_t \gets$ the server matched by $\calC$ at round $t$
            \STATE $\mu_t \gets$ a minimum-cost perfect matching between $S \setminus S^\calC_{t-1}$ and $S \setminus S_{t-1}$, where every server in $S \setminus (S^\calC_{t-1} \cup S_{t-1})$ \\\quad is matched with itself
            \STATE $s_t \gets$ the counterpart of $s^\calC_t$ in $\mu_t$
            \STATE Match $r_t$ with $s_t$
            \STATE $S_t \gets S_{t-1} \cup \{s_t\}$
            \STATE $t \gets t + 1$
        \ENDIF
    \ENDWHILE
  \end{algorithmic}
\end{algorithm}

\paragraph{Analysis}
We now analyze the combination algorithm.

\begin{lemma} \label{lem:app:comb:feas}
    The combination algorithm outputs a perfect matching.
\end{lemma}
\begin{proof}
    Fix any round $t$. 
    Recall that any server in $S \setminus S^\calC_{t-1}$ is matched with a server in $S \setminus S_{t-1}$ in $\mu_t$.
    Note that $s^\calC_t \in S \setminus S^\calC_{t-1}$, implying that $s_t \in S \setminus S_{t-1}$.
\end{proof}

\begin{lemma} \label{lem:app:comb:stepbound}
    For any $t \in \{1, \ldots, n\}$, $d(s_t, s^\calC_t) = \dist (S_{t-1}, S^\calC_{t-1}) - \dist (S_t, S^\calC_t)$.
\end{lemma}
\begin{proof}
    As shown in the proof of \Cref{lem:app:comb:feas}, we have $s_t \in S \setminus S_{t-1}$ and $s^\calC_t \in S \setminus S^\calC_{t-1}$.
    Moreover, since $\mu_t$ is a minimum-cost perfect matching between $S \setminus S_{t-1}$ and $S \setminus S^\calC_{t-1}$ in which $s_t$ and $s^\calC_t$ are matched together, we have
    \begin{align*}
        d(s_t, s^\calC_t) 
        &
        = \dist(S \setminus S_{t-1} , S \setminus S^\calC_{t-1}) - \dist(S \setminus (S_{t-1} \cup \{s_t\}), S \setminus (S^\calC_{t-1} \cup\{s^\calC_t\}))
        \\&
        = \dist(S_{t-1}, S^\calC_{t-1}) - \dist(S_t, S^\calC_t).
    \end{align*}
\end{proof}

\begin{lemma} \label{lem:app:comb:totbound}
    If $\cost(\OPT) \geq 1$, the cost incurred by the combination algorithm is at most $9 \, \min\{\cost(\calA), \cost(\calB)\}$.
\end{lemma}
\begin{proof}    
    If there is only a single phase throughout the execution, the cost incurred by the algorithm is exactly 
    \[
        \cost(\calA) \leq 2 \leq 2 \min \{\cost(\calA), \cost(\calB)\},
    \]
    where the second inequality is due to the assumption that $\cost(\OPT) \geq 1$.

    We now assume that the number of phases is at least 2.
    Let $\ell \geq 2$ denote the last (nonempty) phase.
    For each phase $i$, let $\tau_i$ be the last round of the phase; if a phase $i$ is empty, then we have $\tau_i = \tau_{i-1}$.
    Let $\tau_0 := 0$ for consistency.
    By the definition of the combination algorithm, we have
    \begin{equation}
        \cost(M^\calA_{\tau_i}) \leq 2^i \text{ for odd $i \leq \ell$}
        \; \text{ and } \;
        \cost(M^\calB_{\tau_i}) \leq 2^i \text{ for even $i \leq \ell$. } \label{eq:app:combalg:double}
    \end{equation}

    Moreover, for an odd $i$, the total cost incurred during phase $i$ can be bounded as follows:
    \begin{align}
        \sum_{t = \tau_{i-1} + 1}^{\tau_i} d(s_t, r_t)
        &
        \leq \sum_{t = \tau_{i-1} + 1}^{\tau_i} \Big[ d(s_t, s^\calA_t) + d(s^\calA_t, r_t) \Big] \nonumber
        \\&
        \leq \sum_{t = \tau_{i-1} + 1}^{\tau_i} \Big[ d(s^\calA_t, r_t) + \dist(S_{t-1}, S^\calA_{t-1}) - \dist(S_t, S^\calA_t) \Big] \nonumber
        \\&
        = \sum_{t = \tau_{i-1}+1}^{\tau_i} d(s^\calA_t, r_t) + \dist(S_{\tau_{i-1}}, S^\calA_{\tau_{i-1}}) - \dist(S_{\tau_i}, S^\calA_{\tau_i}) \nonumber
        \\&
        \leq \sum_{t = \tau_{i-1}+1}^{\tau_i} d(s^\calA_t, r_t) + \dist(S^\calA_{\tau_{i-1}}, R_{\tau_{i-1}}) + \dist(S^\calB_{\tau_{i-1}}, R_{\tau_{i-1}}) + \dist(S_{\tau_{i-1}}, S^\calB_{\tau_{i-1}}) - \dist(S_{\tau_i}, S^\calA_{\tau_i}) \nonumber
        \\&
        \leq \sum_{t = \tau_{i-1}+1}^{\tau_i} d(s^\calA_t, r_t) + \cost(M^\calA_{\tau_{i-1}}) + \cost(M^\calB_{\tau_{i-1}}) + \dist(S_{\tau_{i-1}}, S^\calB_{\tau_{i-1}}) - \dist(S_{\tau_i}, S^\calA_{\tau_i}) \nonumber
        \\&
        = \cost(M^\calA_{\tau_i}) + \cost(M^\calB_{\tau_{i-1}}) + \dist(S_{\tau_{i-1}}, S^\calB_{\tau_{i-1}}) - \dist(S_{\tau_i}, S^\calA_{\tau_i}), \label{eq:app:combalg:oddi}
    \end{align}
    where the second inequality is due to \Cref{lem:app:comb:stepbound}.
    For an even $i$, the symmetric derivation gives
    \begin{equation}
        \sum_{t = \tau_{i-1} + 1}^{\tau_i} d(s_t, r_t) \leq \cost(M^\calB_{\tau_i}) + \cost(M^\calA_{\tau_{i-1}}) + \dist(S_{\tau_{i-1}}, S^\calA_{\tau_{i-1}}) - \dist(S_{\tau_i}, S^\calB_{\tau_i}). \textbf{\label{eq:app:combalg:eveni}}
    \end{equation}
    
    For presentational simplicity, let us assume that $\ell$ is odd; a symmetric argument proves the other case.
    Note that, by \Cref{eq:app:combalg:oddi,eq:app:combalg:eveni}, the total cost incurred by the combination algorithm is bounded by
    \begin{align}
        &
        2 \cdot \Bigg[ \sum_{\text{odd $i < \ell$}} \cost(M^\calA_{\tau_i}) + \sum_{\text{even $i < \ell$}} \cost(M^\calB_{\tau_i}) \Bigg] + \cost(M^\calA_{\tau_\ell}) + \dist(S_0, S^\calB_0) - \dist(S_n, S^\calA_n) \nonumber
        \\&
        \quad \leq 2 \sum_{i = 1}^{\ell - 1} 2^i + \cost(\calA) \leq \cost(\calA) + 2^{\ell+1}, \label{eq:app:combalg:finalbnd}
    \end{align}
    where the first inequality follows from \Cref{eq:app:combalg:double} and the fact that $\dist(S_0, S^\calB_0) = \dist(S_n, S^\calA_n) = 0$.

    We now determine the factor.
    If $\cost(\calA) \le \cost(\calB)$, we can further bound \Cref{eq:app:combalg:finalbnd} by 
    \begin{align*}
        \Big(1 + \frac{2^{\ell + 1}}{\cost(\calA)} \Big) \cost(\calA)
        &
        < \Big(1 + \frac{2^{\ell + 1}}{2^{\ell-2}} \Big) \min \{\cost(\calA),\cost(\calB)\}
        = 9 \cdot \min \{\cost(\cA),\cost(\cB)\},
    \end{align*}
    where the inequality is due to the fact that phase $\ell-2$ terminates when $\cost(M^\calA_{\tau_{\ell-2} + 1}) > 2^{\ell-2}$.
    On the other hand, if $\cost(\calA) > \cost(\calB)$, \Cref{eq:app:combalg:finalbnd} can be bounded from above by
    \begin{align*}
        \Big(\frac{\cost(\calA) + 2^{\ell + 1}}{\cost(\calB)} \Big) \cost(\calB)
        &
        < \Big(\frac{2^\ell + 2^{\ell + 1}}{2^{\ell-1}} \Big) \min \{\cost(\calA),\cost(\calB)\}
        = 6 \cdot \min \{\cost(\cA),\cost(\cB)\},
    \end{align*}
    where the inequality follows from \Cref{eq:app:combalg:double} and the fact that phase $\ell-1$ terminates when $\cost(M^\calB_{\tau_{\ell-1}+1}) > 2^{\ell-1}$.
\end{proof}

Following is a simple fact that will be used in the proof of the main theorem.
\begin{fact} \label{fact:prelim:expmin}
    For two random variables $X$ and $Y$, we have $\E[\min\{X, Y\}] \leq \min\{ \E[X], \E[Y]\}$.
\end{fact}
\begin{proof}
    For any realization of $Y$, we have $\min\{X, Y\} \leq X$, implying that $\E[\min\{X, Y\}] \leq \E[X]$.
    A symmetric argument shows $\E[\min\{X, Y\}] \leq \E[Y]$.
\end{proof}

We are now ready to prove the main theorem.
\begin{proof} [Proof of \Cref{thm:prelim:combalgrand}]
    Due to \Cref{lem:app:comb:feas}, the combination algorithm is guaranteed to output a perfect matching.
    From \Cref{lem:app:comb:totbound}, we can infer that the expected cost incurred by the combination algorithm is at most $9 \cdot \E [ \min \{ \cost(\calA), \cost(\calB) \} ]$.
    \Cref{fact:prelim:expmin} then completes the proof.
\end{proof}
\newcommand{\Mon}{M^\mathsf{on}}
\newcommand{\Moff}{M^\mathsf{off}}

\section{Adherence and Strong Competitiveness of Existing Algorithms} \label{app:compalgs}

We prove the adherence and strong competitiveness of the algorithms in \Cref{tab:alg:apply}.

\subsection{Deterministic Algorithms}
All the deterministic algorithms presented in \Cref{tab:alg:apply} can be described in the framework of \citet{nayyar2017input}.

We first need to define the \emph{$\gamma$-net-cost} of an augmenting path.
\begin{definition}[\cite{nayyar2017input}]
    For any matching $M$ and a parameter $\gamma \geq 1$, the \emph{$\gamma$-net-cost} of any augmenting path $P$ is defined as
    \[
        \gamma \cdot \Bigg( \sum_{e \in P \setminus M} d(e) \Bigg) - \sum_{e \in P \cap M} d(e).
    \]
\end{definition}
We can now describe the framework:
It maintains two matchings $\Mon$ and $\Moff$ satisfying:
\begin{itemize}
    \item no reassignments happen in $\Mon$, whereas it happens in $\Moff$;
    \item the matched servers are equivalent between $\Mon$ and $\Moff$.
\end{itemize}
Upon arrival of a request $r_t \in R$, the framework computes an augmenting path $P_t$ from $r_t$ to any unmatched server $s_t$ that achieves the minimum $\gamma$-net cost.
It then augments $\Moff$ along $P_t$, and matches $r_t$ with $s_t$ in $\Mon$.
This is the end of the framework description.

Note that the deterministic $(2n-1)$-competitive algorithm is equivalent with the framework with $\gamma = 1$ \cite{kalyanasundaram1993online, khuller1994line}.
\citet{nayyar2017input} and \citet{raghvendra2018optimal} define $\gamma = 3$.

The proof for adherence of this framework can be found in \citet{nayyar2017input}.
\begin{lemma} \label{lem:app:compalgs:adh}
    This framework is $\gamma$-adherent.
\end{lemma}
\begin{proof}
    It is shown in \cite{nayyar2017input} that the framework always maintains $\Moff$ and a dual solution $y \in \bbR^{S \cup R}_+$ to be \emph{$\gamma$-feasible}; that is, in any round,
    \begin{itemize}
        \item $y(s) + y(r) \leq \gamma \cdot d(s, r)$ for all $(s, r) \in S \times R$;
        \item $y(s) + y(r) = d(s, r)$ for all $(s, r) \in \Moff$;
        \item $y(s) = 0$ and $y(r) = 0$ for all unmatched $s \in S$ and not-yet-arrived $r \in R$.
    \end{itemize}
    Note that $\cost(\Moff) = \sum_{v \in S \cup R} y(v)$ while $\nicefrac{y}{\gamma}$ is feasible to the dual of the weighted bipartite matching LP relaxation.
    This completes the lemma.
\end{proof}

The strong competitiveness of this algorithm can be easily derived from the adherence and determinism.
\begin{lemma}
    For any $\rho : \bbN \to \bbR$, if this framework is $\rho$-competitive, it is strongly $\gamma \rho$-competitive.
\end{lemma}
\begin{proof}
    For any round $t$, let $S_t$ denote the servers matched by this framework up to and including round $t$.
    Consider an auxiliary instance $(S, d, S_t, R_t)$.
    The framework is guaranteed to output a perfect matching between $S_t$ and $R_t$ of cost at most $\rho(t) \cdot \dist(S_t, R_t)$.
    Observe that the execution of this framework given $(V, d, S, R)$ up to round $t$ is equivalent with that given $(V, d, S_t, R_t)$ until termination.
    The lemma then immediately follows from \Cref{lem:app:compalgs:adh}.
\end{proof}

The above two lemmas together complete the proof of adherence and strong competitiveness of the deterministic algorithms in \Cref{tab:alg:apply}.

\subsection{Randomized Algorithm}
As shown in \citet{meyerson2006randomized} and \citet{bansal2014randomized}, it is without loss of generality to assume that the vertices of the given metric spaces are precisely the servers (i.e., $V = S$) because one can easily modify any general metric space into such one at a constant factor in competitive ratios.
\citet{bansal2014randomized} presented an algorithm on 2-HSTs.
\begin{definition}[HST] \label{def:app:hst}
    Given $\alpha \geq 1$, an $\alpha$-hierarchical well-separated tree ($\alpha$-HST) is a rooted tree $T = (U, E)$ along with a length $\ell : E \to \R_+$ satisfying
    \begin{itemize}
        \item for each internal node $u \in U$, the length between $u$ and any child of $u$ is identical, i.e., $\ell(u, c_1) = \ell(u, c_2)$ for any two children $c_1$ and $c_2$ of $u$;
        \item for any internal node $u \in U$, the length between $u$ and its parent $p$ is exactly $\alpha$ times the length between $u$ and its child $c$, i.e., $\ell(u, p) = \alpha \cdot \ell (u, c)$.
    \end{itemize}
\end{definition}

It is well-known that any metric space can be embedded into an HST with a logarithmic expected distortion:
\begin{theorem}[\cite{fakcharoenphol2004tight}] \label{thm:app:hst}
    Given any metric $(V, d)$ with $n$ vertices, there exists a randomized algorithm producing 2-HST $(T, \ell)$ such that
    \begin{itemize}
        \item the leaf nodes of $T$ are corresponding to $V$;
        \item for any $u, v \in V$, the distance between $u$ and $v$ in $(T, \ell)$ is at least $d(u, v)$;
        \item for any $u, v \in V$, the expected distance between $u$ and $v$ in $(T, \ell)$ is bounded from above by $O(\log n) \cdot d(u, v)$.
    \end{itemize}
    The factor of $d(u, v)$ in the last property is called the (expected) distortion of $(T, \ell)$.
\end{theorem}

We now describe the algorithm of \citet{bansal2014randomized}.
Given a 2-HST $(T = (U, E), \ell)$ with $n$ leaves, for every leaf node $u \in U$ and height $h \geq 0$, let $T(u, h)$ denote the set of leaf nodes that belong to the subtree of height $h$ containing $u$, and let $N(u, h) := T(u, h) \setminus T(u, h-1)$ (and $N(u, 0) := T(u, 0)$ for $h = 0$).
The algorithm maintains two matchings $\Mon$ and $\Moff$ such that
\begin{itemize}
    \item no reassignment happen in $\Mon$, but it happens in $\Moff$;
    \item the matched servers are equivalent between $\Mon$ and $\Moff$.
\end{itemize}
Throughout the execution, the algorithm maintains the \emph{level} of every server $s \in S$, denoted by $L(s)$, defined as the minimum height of a subtree containing both $s$ and its matched request in $\Moff$ (and $L(s) := \infty$ if unmatched).
Accordingly, it also maintains the level of every request $r$ that has arrived so far, defined as the level of its matched server in $\Moff$.

Upon arrival of a request $r_t \in R$, the algorithm modifies $\Moff$ by the following recursive procedure:
Let $r$ be the request to be (re)assigned, initially $r := r_t$ and $L(r) := 0$.
It finds the lowest height $h \geq L(r)$ where there is a server $s \in N(r, h)$ whose current level is higher than $h$, i.e., $L(s) > h$.
It then chooses a server $s'$ uniformly at random among the servers $s$ in $N(r, h)$ with $L(s) > h$, and (re)assign $s'$ to $r$ in $\Moff$ with setting $L(s') = L(r) = h$.
If $s'$ was matched with another request $r'$ in $\Moff$, the algorithm recursively reassigns $r'$ using the procedure.
Otherwise, if $s'$ was unmatched, it then matches $r_t$ with $s'$ in $\Mon$ and proceeds to the next round.

Its exact adherence is immediate from the proofs of Lemma 4.1 and 4.3 in \citet{bansal2014randomized}, so we omit its proof here.
\begin{lemma}
    This algorithm is 1-adherent.
\end{lemma}

The strong competitiveness can be shown by a modest adaptation of the proof of Lemma 4.4 in \citet{bansal2014randomized}.
\begin{lemma}
    For any round $t$, the expected cost of the matching maintained by the algorithm at round $t$ is at most $O(\log t) \cdot \min_{M \in \calM_t} \cost(M)$.
\end{lemma}
\begin{proof}
    We follow the same line of the proof of Lemma 4.4 in \citet{bansal2014randomized}.
    Consider any request $r \in R_t$, and let $L(r)$ denote the last level of $r$ at termination.
    It suffices to show that, for every height $h \in \{0, \ldots, L(r)\}$, the expected number of times $r$ is (re)assigned to servers in $N(r, h)$ is bounded by $O(\log t)$.
    Assume a round during the execution up to and including round $t$ where $r$ is (re)assigned to a server in $N(r, h)$, and let $W \subseteq N(r, h)$ be the servers in $N(r, h)$ with $L(s) > h$.
    Observe that $W$ never increases throughout the execution.
    Moreover, since $r$ can be reassigned only when the newly arrived request $r'$ is closer to the server currently matched to $r$, all servers in $W$ are also potential servers to be reassigned upon arrival of $r'$, implying that the probability of the reassignment of $r$ is at most $\nicefrac{1}{|W|}$ for the current $W$.
    Finally, if $r$ is reassigned, $W$ decreases by 1.
    Since we only consider the execution up to round $t$, we can therefore conclude that the expected number of times $r$ is reassigned to $N(r, h)$ is at most
    \[
        \sum_{i = 1}^{\min\{t, |W|\}} \frac{1}{|W| - i + 1} \leq H_t,
    \]
    where $W$ here denotes the snapshot of $W$ when $r$ is (re)assigned to a server in $N(r, h)$ for the first time and $H_t$ is the $t$-th Harmonic number.
    The rest of the argument is identical to the proof of Lemma 4.4 in \citet{bansal2014randomized}.
\end{proof}
\newcommand{\calT}{\mathcal{T}}

\section{Deferred Lower Bound Analysis} \label{app:hard}
This appendix is devoted to proving the lower bound results in \Cref{sec:hard}. 
In particular, we present adversaries that bound from below the (expected) cost any deterministic (or randomized) algorithm must incur even if the oracle always provide accurate predictions.
For the sake of analysis, we define problem setting $\calT_n$ as a setting that has a metric space induced from a star graph with $n$ leaf nodes and one center node, and a set of $n$ servers located at $n$ leaf nodes.

\subsection{Deterministic Algorithms}
\begin{proof} [Proof of \Cref{thm:hard:det}]
Consider a problem setting $\calT_n$.
As $\cA$ is deterministic, we can assume an adaptive adversary.
Before describing the adversary, we first introduce a couple of definitions.
In what follows, the adversary will reveal at most one request at each leaf; we say a server is \emph{requested} if a request is revealed at the same location as the server.
Moreover, upon a prediction query from $\cA$, the adversary will construct a prediction from the servers matched by $\cA$ so far plus one server unmatched by $\cA$ then; we say a server is \emph{predicted} if it has been an unmatched server of the prediction queried in any round.

We can now describe the adversary: In every round $t \in \{1, \ldots, n\}$, if there exists a server $s \in S$ that is matched by $\cA$, unpredicted, and unrequested, the adversary reveals a request at the location of $s$; otherwise, if there does not exist such a server, the adversary reveals a request at the center.
For the construction of prediction $P_t$ in round $t$ (in case where $\cA$ makes a query), if there exists a predicted, unmatched server $p$, the adversary sets $P_t := S_{t-1} \cup \{p\}$, where $S_{t-1}$ is a set of the servers matched by $\cA$ at the beginning of round $t$; otherwise, if such a server $p$ does not exist, the adversary selects any unmatched server as $p$ and defines $P_t := S_{t-1} \cup \{p\}$.

Observe that the adversary guarantees that at most one predicted, unmatched server $p$ exists throughout the execution while the adversary reveals a new request at the center only when $p$ got matched by $\cA$ in the previous round.
This observation implies that the number of the requests located at the center is at most $B + 1$, and hence, $\cost (\OPT) \leq B+1$.
On the other hand, since $\cA$ incurs two units to serve every request located at a leaf, we can derive \[ \cost (\cA) \geq 2n - B - 1 \geq \Big( \frac{2n}{B+1} - 1 \Big) \cost (\OPT).\]
\end{proof}

\begin{proof} [Proof of \Cref{thm:hard:detwsq}]
    Consider a problem setting $\calT_n$.
    The adversary first sends a request to the center in the first round.
    In each subsequent round, the adversary issues a request to the location of the server most recently matched by the online algorithm. 
    This process continues until the $k$-th round, at which point the algorithm receives a prediction revealing which server is matched to the center in the offline solution. 
    From that point on, the adversary can no longer continue this strategy.

    In the remaining rounds, the adversary issues requests to arbitrary leaf nodes that have not yet received any request.

    This construction ensures that the online algorithm incurs a cost of at least $2k-1$, while the offline algorithm incurs a cost of exactly $1$. Consequently, the competitive ratio is at least $2k-1$.
\end{proof}

\subsection{Randomized Algorithms}

\begin{lemma}[see, e.g., \cite{meyerson2006randomized}] \label{lem:hard:randcomp}
    For the problem setting $\calT_n$, any algorithm can incur cost of at least $2 H_n - 1$ in expectation, while the minimum cost of a perfect matching with hindsight is $1$.
\end{lemma}
\begin{proof}
    The strategy for the adversary is as follows:
    \begin{enumerate}
        \item In the first round, send a request to the center node.
        \item For the remaining rounds, uniformly at random pick a node that has not yet received any request and then send a request to that node.
    \end{enumerate}
    For each round $t \in \{2, \ldots, n\}$, the probability that the server co-located with $r_t$ that has already been matched by the algorithm is $\frac{1}{n-t+2}$.
    Hence, the total expected cost is $1 + 2 \sum_{t = 2}^n \frac{1}{n-t+2} = 2H_n - 1$.
    It is easy to see that the minimum cost of a perfect matching with hindsight is always $1$.
\end{proof}

We are now ready to prove \Cref{thm:hard:randbnp}, our lower bound under the regime of bounded number of predictions.

\begin{proof}[Proof of \Cref{thm:hard:randbnp}]

Consider the problem setting $\calT_n$.
The adversary follows the strategy in \Cref{lem:hard:randcomp} until the algorithm makes a query to the prediction oracle.
Because the prediction is perfectly accurate, the algorithm then learns which server is matched to the first request at the center and can subsequently match the most recent request to that server.

To ensure that the algorithm continues to incur nonzero cost, the adversary restarts the strategy in \Cref{lem:hard:randcomp} from the beginning whenever the algorithm makes a query to the prediction oracle, and repeats this process until the next query is made. 
This strategy is possible since the algorithm deterministically makes the queries to the oracle.
Note that any server that has appeared in a prediction is treated as having received a request, since it is either destined to be matched to a request at the center node or already gets a request.

Let $q$ be the number of queries the algorithm makes to the oracle excluding one made in the $n$-th round.
Then, $q+1$ requests are sent to the center.
Hence, the offline algorithm pays cost of $q+1$. However, the online algorithm pays a cost of one for the request at the center and expected cost of $\frac{2}{j}$ for an event that a request is sent to an already occupied server in round $n-j+1$.
Thus, the algorithm incurs a cost of at least $q+1+\sum_{j=q+2}^{n} \frac{2}{j}=q+1+2(H_n-H_{q+1})$.
Since $q\le B$, the competitive ratio is at least $1+\frac{2(H_k-H_{q+1})}{q+1}\ge 1+\frac{2(H_k-H_{B+1})}{B+1}$.
\end{proof}

We now prove \Cref{thm:hard:randwsq}.

\begin{proof} [Proof of \Cref{thm:hard:randwsq}]
    Consider the problem setting $\calT_n$.
    Suppose that $n$ is a multiple of $k$.
    For the first $n-k$ rounds, the adversary sends the requests to only leaf nodes.
    Since $n$ is a multiple of $k$, the prediction is given to the algorithm in the $(n-k)$-th round.
    After that, in the $(n-k+1)$-th round, the adversary sends a request to the center node and follows the strategy in \Cref{lem:hard:randcomp}.
    Because both the offline and online algorithms incur zero cost during the first $n-k$ rounds, and because the remaining free servers together with the center node form a  setting $\calT_k$ for the last $k$ rounds, this yields an expected lower bound of $2H_k-1$.
\end{proof}

As one can construct a 2-HST with constant distortion out of a star metric, the above lower bound complements that the algorithm from \Cref{thm:alg:randhst} is best-possible (up to a constant factor) for $2$-HSTs under the regime of well-separated queries to the oracle.

\subsection{Stronger Lower Bounds for Well-Separated Queries Model}

\begin{proof}[Proof Sketch]
Observe that, in the well-separated queries model, the adversary can modify predictions in order to mislead the algorithm. In particular, for the problem instance $\calT_n$, the adversary can employ the same strategies as in \Cref{thm:hard:detwsq,thm:hard:randwsq}, but over a longer duration, by concealing the server that is matched to the center request in the offline solution.

This concealment can be achieved by altering only a single server in each prediction. Consequently, the adversary can extend the execution of the strategy by additional $k$ rounds for every error budget of $2$.

\end{proof}
We therefore obtain the following theorems:

\begin{theorem}
    There exists an instance such that deterministic algorithm making well-separated queries with separation parameter $k$ incurs cost of $2k-1 + k\eta$, where $\eta = \sum_{t=1}^{\lfloor \nicefrac{n}{k} \rfloor} \dist(\Opt_{tk}, P_{tk})\le 2(\nicefrac{n}{k}-1)$, while the minimum cost of a perfect matching with hindsight is 1.
\end{theorem}

\begin{theorem}
    There exists an instance such that randomized algorithm making well-separated queries with separation parameter $k$ incurs cost of $\Omega(\log (k+k\eta))$, where $\eta = \sum_{t=1}^{\lfloor \nicefrac{n}{k} \rfloor} \dist(\Opt_{tk}, P_{tk})\le 2(\nicefrac{n}{k}-1)$, while the minimum cost of a perfect matching with hindsight is 1.
\end{theorem}
\section{Other Deferred Proofs} \label{app:dproofs}

\begin{proof}[Proof of \Cref{fact:prelim}]
    The first half of the statement is easy to see since, in a minimum-cost perfect matching between $A \uplus C$ and $B \uplus C$, we can assume without loss of generality that every $v \in C$ is matched with $v$ itself.
    
    The second half of the statement is due to the triangle inequality.
    More precisely, let $M_{AC}$ and $M_{CB}$ be minimum-cost perfect matchings between $A$ \& $C$ and between $C$ \& $B$, respectively.
    Now consider a perfect matching $M_{AB}$ between $A$ and $B$ defined as follows: every $u \in A$ is matched with the counterpart $v \in B$ in $M_{CB}$ of the counterpart $w \in C$ in $M_{AC}$ of $u$, i.e., $(u, w) \in M_{AC}$ and $(w, v) \in M_{CB}$.
    Note that, by the triangle inequality, $d(u, v) \leq d(u, w) + d(w, v)$.
    We can thus derive 
    \[
        \dist(A, B) \leq \cost(M_{AB}) \leq \cost(M_{AC}) + \cost(M_{CB}) = \dist(A,C) + \dist(C,B).
    \]
\end{proof}

\begin{proof}[Proof of \Cref{lem:alg:ftp}]
    We claim that, for every round $t$,
    \begin{equation*}
        d(s_t, r_t) \leq \dist(P_t, P_{t-1} \cup \{r_t\}) + \dist(S_{t-1}, P_{t-1}) - \dist(S_t, P_t).
    \end{equation*}
    Note that this claim immediately proves the lemma since the last two terms in the right-hand side telescope, and $\dist(S_0, P_0) = \dist(S_n, P_n) = 0$.

    First, by the construction of $p_t$ and $s_t$, we can derive
    \begin{align}
        d(p_t, r_t)
        &
        = \dist(P_t, P_{t-1} \cup \{r_t\}) - \dist(P_t \setminus \{p_t\}, P_{t-1}) \nonumber
        \\&
        = \dist(P_t, P_{t-1} \cup \{r_t\}) - \dist(P_t , P_{t-1} \cup \{p_{t}\}); \label{eq:app:dproof:pr}
        \\
        d(s_t, p_t) 
        & 
        = \dist(S \setminus (S_{t-1} \setminus P_{t-1}), S \setminus (P_{t-1} \setminus S_{t-1})) - \dist(S \setminus (S_{t-1} \setminus P_{t-1}) \setminus \{s_t\}, S \setminus (P_{t-1} \setminus S_{t-1}) \setminus \{p_t\}) \nonumber
        \\&
        = \dist(S_{t-1}, P_{t-1}) - \dist(S_t, P_{t-1} \cup \{p_t\}). \label{eq:app:dproof:sp}
    \end{align}
    We thus have
    \begin{align*}
        d(s_t, r_t)
        &
        \leq d(p_t, r_t) + d(s_t, p_t)
        \\&
        = \dist(P_t, P_{t-1} \cup \{r_t\}) - \dist(P_t , P_{t-1} \cup \{p_{t}\}) + \dist(S_{t-1}, P_{t-1}) - \dist(S_t, P_{t-1} \cup \{p_t\})
        \\&
        \leq \dist(P_t, P_{t-1} \cup \{r_t\}) + \dist(S_{t-1}, P_{t-1}) - \dist(S_t, P_t),
    \end{align*}
    where the first inequality is due to the triangle inequality, the equality comes from \Cref{eq:app:dproof:pr,eq:app:dproof:sp}, and the second inequality follows from the fact that $d(S_t, P_t) \leq d(S_t, P_{t-1} \cup \{p_t\}) + \dist(P_{t-1} \cup \{p_t\}, P_t)$ due to \Cref{fact:prelim}.
    This completes the proof of the claim.
\end{proof}
\section{Generation of a \taxiinst{} Instance} \label{app:taxigen}

A \taxiinst{} instance is generated using data from October~19,~2023 in the Chicago Taxi Trips (2013--2023) dataset~\cite{chicago_taxi_trips}.
The following fields are used in the instance construction:
\begin{itemize}
    \item \texttt{Trip Start Timestamp},
    \item \texttt{Trip End Timestamp},
    \item \texttt{Dropoff Centroid Longitude},
    \item \texttt{Dropoff Centroid Latitude},
    \item \texttt{Pickup Centroid Longitude},
    \item \texttt{Pickup Centroid Latitude}.
\end{itemize}

We partition the time span of October~19,~2023 into 5-minute intervals and uniformly sample one time point from this set.
Given a sampled time point $\tau$, we identify the latest 100 entries whose \texttt{Trip End Timestamp} does not exceed $\tau$ and treat them as servers.
Similarly, we identify the earliest 100 entries whose \texttt{Trip Start Timestamp} is not earlier than $\tau$ and treat them as requests arriving sequentially.
If fewer than 100 servers or 100 requests are available, we discard the instance and repeat the procedure with an independently sampled time point.
The distance among servers and requests are defined as the Manhattan distance using the servers'
(\texttt{Dropoff Centroid Longitude}, \texttt{Dropoff Centroid Latitude})
and the requests'
(\texttt{Pickup Centroid Longitude}, \texttt{Pickup Centroid Latitude}).

\end{document}